\documentclass[submission,copyright,creativecommons,fleqn]{eptcs}
\providecommand{\event}{EXPRESS/SOS 2020} 
\usepackage{breakurl}             
\usepackage{underscore}           

\usepackage{amsmath}
\usepackage{amssymb}
\usepackage{amsthm}
\usepackage{procalg}
\usepackage{tikz}
\usetikzlibrary{arrows,automata,shapes,positioning}

\theoremstyle{plain}
\newtheorem{thm}{Theorem}

\newtheorem{prop}[thm]{Proposition}

\theoremstyle{definition}
\newtheorem{defn}[thm]{Definition}
\newtheorem{exmp}[thm]{Example}

\newcommand{\Acts}{\ensuremath{\mathalpha{\Act{}_{\silent}}}}
\newcommand{\act}[1][\alpha]{\ensuremath{\mathalpha{#1}}}
  \newcommand{\acta}[1][]{\ensuremath{\act[a_{#1}]}}
  
\newcommand{\silent}{\act[\tau]}
  \newcommand{\States}[1][S]{\ensuremath{\mathalpha{#1}}}
  \newcommand{\state}[1]{\ensuremath{\mathalpha{#1}}}
    \newcommand{\states}[1][]{\state{s_{#1}}}
    \newcommand{\statet}[1][]{\state{t_{#1}}}

  \newcommand{\brelsym}[1][\mathcal{R}]{\ensuremath{#1}}
  \newcommand{\brel}[1][{\brelsym}]{\ensuremath{\mathrel{#1}}}
\newcommand{\opt}[1]{\mbox{\tiny\rm(}#1\mbox{\tiny\rm)}} 
\newcommand{\plat}[1]{\raisebox{0pt}[0pt][0pt]{#1}}      
   \newdimen\boxwdplusemdimen
   \def\arrow#1{{
     \boxwdplusemdimen=1em%
     \setbox0=\hbox{$\scriptstyle#1$}%
     \advance\boxwdplusemdimen by \wd0\relax%
     \ifdim\boxwdplusemdimen<16.11119pt%
       \boxwdplusemdimen=16.11119pt%
     \fi%
     \buildrel{#1}\over%
       {\setbox1=\hbox to \boxwdplusemdimen{\rightarrowfill}%
     \ht1=0.3em\relax\box1}%
   }}
   \makeatletter
   \def\twoheadrightarrowfill{$\m@th\smash-\mkern-7mu%
     \cleaders\hbox{$\mkern-2mu\smash-\mkern-2mu$}\hfill
     \mkern-7mu\mathord\twoheadrightarrow$}
   \makeatother
   \def\darrow#1{{
     \boxwdplusemdimen=1em%
     \setbox0=\hbox{$\scriptstyle#1$}%
     \advance\boxwdplusemdimen by \wd0\relax%
     \ifdim\boxwdplusemdimen<16.11119pt%
       \boxwdplusemdimen=16.11119pt%
     \fi%
     \buildrel{#1}\over%
       {\setbox1=\hbox to \boxwdplusemdimen{\twoheadrightarrowfill}%
     \ht1=0.3em\relax\box1}%
   }}
   \def\plarrow#1{{
     \boxwdplusemdimen=1em%
     \setbox0=\hbox{$\scriptstyle#1$}%
     \advance\boxwdplusemdimen by \wd0\relax%
     \ifdim\boxwdplusemdimen<16.11119pt%
       \boxwdplusemdimen=16.11119pt%
     \fi%
     \buildrel{#1}\over%
       {\setbox1=\hbox to \boxwdplusemdimen{\rightarrowfill}%
     \ht1=0.3em\relax\box1^{\scriptstyle +}}%
   }}
  \newcommand{\stepsym}{\ensuremath{\mathalpha{\longrightarrow}}}

  \newcommand{\ssteps}{\ensuremath{\mathbin{\darrow{}}}}
  \newcommand{\sstepssym}{\ensuremath{\mathbin{\twoheadrightarrow}}}
  \newcommand{\plusstepssym}{\ensuremath{\mathbin{\rightarrow^{+}}}}
  \newcommand{\plussteps}{\ensuremath{\mathbin{\plarrow{}}}}
  \newcommand{\wbbisimd}{%
    \setbox0=\hbox{\kern-.1ex{$\leftrightarrow$}\kern-.1ex}
    \setbox1=\vbox{\hbox{\raise .1ex \box0}\hrule}%
    \ensuremath{\mathrel{\hbox{\kern.1ex\box1\kern.1ex}^{\Delta_3}_{b}}}
  }
  \newcommand{\wbbisimdtwo}{%
    \setbox0=\hbox{\kern-.1ex{$\leftrightarrow$}\kern-.1ex}
    \setbox1=\vbox{\hbox{\raise .1ex \box0}\hrule}%
    \ensuremath{\mathrel{\hbox{\kern.1ex\box1\kern.1ex}^{\Delta_2}_{b}}}
  }
  \newcommand{\onestepbbisimd}{%
    \setbox0=\hbox{\kern-.1ex{$\leftrightarrow$}\kern-.1ex}
    \setbox1=\vbox{\hbox{\raise .1ex \box0}\hrule}%
    \ensuremath{\mathrel{\hbox{\kern.1ex\box1\kern.1ex}^{\Delta_1}_{b}}}
  }
  \newcommand{\bbisimd}{%
    \setbox0=\hbox{\kern-.1ex{$\leftrightarrow$}\kern-.1ex}
    \setbox1=\vbox{\hbox{\raise .1ex \box0}\hrule}%
    \ensuremath{\mathrel{\hbox{\kern.1ex\box1\kern.1ex}^{\Delta}_{b}}}
  }
  \newcommand{\bbisimdX}{%
    \setbox0=\hbox{\kern-.1ex{$\leftrightarrow$}\kern-.1ex}
    \setbox1=\vbox{\hbox{\raise .1ex \box0}\hrule}%
    \ensuremath{\mathrel{\hbox{\kern.1ex\box1\kern.1ex}^{\Delta}_{bX}}}
  }
  \newcommand{\rbbisimd}{%
    \setbox0=\hbox{\kern-.1ex{$\leftrightarrow$}\kern-.1ex}
    \setbox1=\vbox{\hbox{\raise .1ex \box0}\hrule}%
    \ensuremath{\mathrel{\hbox{\kern.1ex\box1\kern.1ex}^{\Delta}_{rb}}}
  }
  \newcommand{\rbbisimdX}{%
    \setbox0=\hbox{\kern-.1ex{$\leftrightarrow$}\kern-.1ex}
    \setbox1=\vbox{\hbox{\raise .1ex \box0}\hrule}%
    \ensuremath{\mathrel{\hbox{\kern.1ex\box1\kern.1ex}^{\Delta}_{rbX}}}
  }
  \newcommand{\rbbisimdzero}{%
    \setbox0=\hbox{\kern-.1ex{$\leftrightarrow$}\kern-.1ex}
    \setbox1=\vbox{\hbox{\raise .1ex \box0}\hrule}%
    \ensuremath{\mathrel{\hbox{\kern.1ex\box1\kern.1ex}^{\Delta_0}_{rb}}}
  }

  \newcommand{\N}{\ensuremath{\mathalpha{\omega}}}

  \newcommand{\may}[1]{\langle #1 \rangle}  
  \newcommand{\must}[1]{[ #1 ]}
  \newcommand{\jb}[1]{\mathbin{#1}}
  \newcommand{\Div}{\mathalpha{\Delta}}

\title{Divergence-Preserving Branching Bisimilarity}
\author{Bas Luttik
\institute{Eindhoven University of Technology\\ The Netherlands}
\email{s.p.luttik@tue.nl}
}
\def\titlerunning{Divergence-Preserving Branching Bisimilarity}
\def\authorrunning{B. Luttik}
\begin{document}
\maketitle

\begin{abstract}
  This note considers the notion of divergence-preserving branching
  bisimilarity. It briefly surveys results pertaining to the notion
  that have been obtained in the past one-and-a-half decade, discusses
  its role in the study of expressiveness of process calculi, and
  concludes with some suggestions for future work.
\end{abstract}

\section{Introduction}

\emph{Branching bisimilarity} was proposed by van Glabbeek and
Weijland as an upgrade of (strong) bisimilarity that facilitates
abstraction from internal activity \cite{GW96}. It preserves the
branching structure of processes more strictly than Milner's
\emph{observation equivalence} \cite{Mil80}, which, according to van
Glabbeek and Weijland, makes it, e.g., better suited for verification purposes. A case
in point is the argument by Graf and Sifakis that there is no
temporal logic with an \emph{eventually} operator that is adequate for
observation equivalence in the sense that two processes satisfy the
same formulas if, and only if, they are observationally equivalent
\cite{GS87}.  The crux is that observation
equivalence insufficiently takes into account the intermediate states
of an internal computation. Indeed, branching bisimilarity requires a
stronger correspondence between the intermediate states of an internal
computation.

Branching bisimilarity is also not compatible with a
temporal logic that includes an eventually operator, because it
abstracts to some extent from \emph{divergence} (i.e., infinite
internal computations). Thus, a further upgrade is necessary, removing
the abstraction from divergence. De Nicola and Vaandrager show that
\emph{divergence-sensitive} branching bisimilarity coincides with the equivalence
induced by satisfaction of formulas of the temporal logic
CTL$^{*}_{-X}$ \cite{DV95}. (CTL$^{*}$ \cite{EH86} is an expressive
state-based logic that includes both linear time and branching time
modalities; CTL$^{*}_{-X}$ refers to the variant of CTL$^{*}$ obtained
by omitting the next-state modality, which is incompatible with
abstraction from internal activity.)

Divergence-sensitive branching bisimilarity still has one drawback when it
comes to verification: it identifies deadlock and livelock and, as an
immediate consequence, is not compatible with parallel composition. It
turns out that the notion of \emph{divergence-preserving} branching
  bisimilarity\footnote{For stylistic reasons we prefer the term
  ``divergence-preserving branching bisimilarity'' over ``branching
  bisimilarity with explicit divergence'', which is used in earlier
  articles on the topic.}, which is the topic of this note, has all the right
properties: it is the coarsest equivalence that is compatible with parallel
composition, preserves CTL$^{*}_{-X}$ formulas, and distinguishes
deadlock and livelock \cite{GLT09b}. Moreover, on finite processes
divergence-preserving branching bisimilarity can be decided
efficiently \cite{GJKW17}.

In \cite{GW96}, a coloured-trace characterisation of
divergence-preserving branching bisimilarity is provided.  In
\cite{Gla93a}, relational and modal characterisations of the notion
are given. For some time it was simply assumed that these three
characterisations of the notion coincide, but this was only proved in
\cite{GLT09a}. To establish that the relational characterisation
coincides with the coloured-trace and modal characterisations, it
needs to be proved that the relational characterisation yields an
equivalence relation that satisfies the so-called stuttering property,
and this is surprisingly involved. A similar phenomenon
is observed in the proof that a rooted version
of divergence-preserving branching bisimilarity is compatible with the
recursion construct $\mu X.\_$ \cite{GLS20}. Due to the divergence condition,
Milner's ingenious argument in \cite{Mil89h} that strong bisimilarity
is compatible with recursion required several novel twists.

In this note we shall give a survey of results pertaining to
divergence-preserving branching bisimilarity that
were obtained in the past one-and-a-half decade. in
Section~\ref{sec:relchar} we shall present and discuss a relational
characterisation of the notion. In Section~\ref{sec:modchar} we
comment on modal characterisations of the notion, and discuss the
relationship with the temporal logic CTL$^{*}_{-X}$. In
Section~\ref{sec:congruence} we briefly discuss to what extent the
notion is compatible with familiar process algebraic operators. In
Section~\ref{sec:applications}, we explain how it plays a role in
expressiveness results. In Section~\ref{sec:conclusions} we arrive at some
conclusions and mention some ideas for future work.

\section{Relational characterisation} \label{sec:relchar}

  We presuppose a set
  $\Acts$ of \emph{actions} including a special element
  $\silent$, and we presuppose a \emph{labelled transition
    system} $(\States,\stepsym)$ with labels from $\Acts$, i.e.,
  $\States$ is a set of \emph{states} and
  $\stepsym\subseteq\States\times\Acts\times\States$ is a
  \emph{transition relation} on $\States$.
  Let $\states,\state{s'}\in\States$ and $\act\in\Acts$; we write
  $\states\step{\act}\state{s'}$ for
  $(\states,\act,\state{s'})\in\stepsym$ and we abbreviate the
  statement `$\states\step{\act}\state{s'}$ or ($\act=\silent$ and
  $\states=\state{s'}$)' by
  \plat{$\states\step{\opt{\act}}\state{s'}$}.  We denote by
  $\plusstepssym$ the transitive closure of the binary relation
  $\step{\silent}$, and by $\sstepssym$ its reflexive-transitive
  closure.
  A \emph{process} is given by a state $s$ in a labelled transition
  system, and encompasses all the states and transitions reachable
  from $s$.
  
 \begin{defn} \label{def:bbisimd}
  A symmetric binary relation $\brelsym$ on $\States$ is a
  \emph{branching bisimulation} if it satisfies the following
  condition for all $\states,\statet\in\States$ and $\act{}\in\Acts$:
  \begin{enumerate}\itemsep 0pt
  \renewcommand{\theenumi}{T}
  \renewcommand{\labelenumi}{(\theenumi)}
  \item \label{cnd:stepsim}
    if $\states\brel\statet$ and $\states\step{\act}\states'$ for
    some state $\states'$, then there exist states $\statet'$ and $\statet''$
    such that
      \plat{$\statet\ssteps{}\statet''\step{\opt{\act}}\statet'$},
      $\states\brel\statet''$
    and
      $\states'\brel\statet'$.
  \end{enumerate}
  We say that a branching bisimulation $\brelsym$ \emph{preserves
    (internal) divergence}
  if it satisfies the following condition for all $\states,\statet\in\States$:
  \begin{enumerate}
  \renewcommand{\labelenumi}{(\theenumi)}
  \renewcommand{\theenumi}{D}
  \item \label{cnd:divsim}
     if $\states\brel\statet$ and there is an infinite sequence of
     states
         $(\states[k])_{k\in\N}$
       such that
          $\states=\states[0]$,
          $\states[k]\step{\silent}\states[k+1]$
       and
          $\states[k]\brel\statet$ for all $k\in\N$,
      then there is a state $\statet'$
       such that
          $\statet\plussteps\statet'$,
       and
          $\states[k]\brel\statet'$ for some $k\in\N$.
  \end{enumerate}
  States $\states$ and $\statet$ are \emph{divergence-preserving
    branching bisimilar} (notation: $\states\bbisimd\statet$) if there is a
  divergence-preserving branching bisimulation $\brelsym$ such that
  $\states\brel\statet$. 
\end{defn}
The divergence condition \eqref{cnd:divsim} in the definition above is
slightly weaker than the divergence condition used in the relation
characterisation of divergence-preserving branching bisimilarity
presented in \cite{Gla93a}, which actually requires that $\statet$
admits an infinite sequence of $\silent$-transitions and every state
on this sequence is related to some state on the infinite sequence of
$\silent$-transitions from $\states$. Nevertheless, as is established in
\cite{GLT09a}, the notion of divergence-preserving branching
bisimilarity defined here is equivalent to the one defined in
\cite{Gla93a}. In \cite{GLT09a} it is also proved that $\bbisimd$ is
an equivalence, that the relation $\bbisimd$ is itself a
divergence-preserving branching bisimulation, and that it satisfies
the so-called \emph{stuttering property}: if
$\statet[0]\step{\silent}\cdots\statet[n]$,
$\states\bbisimd\statet[0]$ and $\states\bbisimd\statet[n]$, then
$\states\bbisimd\statet[i]$ for all $0\leq i \leq n$.

Let us say that a state $\states$ is \emph{divergent} if there exists an infinite
sequence of states $(\states[k])_{k\in\N}$ such that
$\states=\states[0]$ and $\states[k]\step{\silent}\states[k+1]$ for
all $k\in\N$. It is a straightforward consequence of the
definition that divergence-preserving branching bisimilarity relates
divergent states to divergent states only, i.e., that we have the following proposition.
\begin{prop}
  If $\states\bbisimd\statet$, then $\states$ is divergent only if $\statet$ is divergent.
\end{prop}
\begin{proof}
  Suppose that $\states\bbisimd\statet$ and $\states$ is divergent.
  Then there exists an infinite sequence of states $(\states[k])_{k\in\N}$ such that
$\states=\states[0]$ and $\states[k]\step{\silent}\states[k+1]$ for
all $k\in\N$. We inductively construct an infinite sequence of states
$(\statet[\ell])_{\ell\in\N}$ such that $\statet=\statet[0]$,
$\statet[\ell]\step{\silent}\statet[\ell+1]$, together with a mapping
$\sigma:\N\rightarrow\N$ such that
$\states[\sigma(\ell)]\bbisimd\statet[\ell]$ for all $\ell\in\N$;
\begin{itemize}
\item We define $\statet[0]=\statet$ and $\sigma(0)=0$; note that $\states[\sigma(0)]=\states\bbisimd\statet=\statet[0]$.
\item Suppose that the sequence $(\statet[\ell])_{\ell\in\N}$ and the
  mapping $\sigma$ have been defined up to $\ell$. Then, in
  particular, $\states[\sigma(\ell)]\bbisimd\statet[\ell]$.
  We distinguish two cases:
  
  If $\states[\sigma(\ell)+k]\bbisimd\statet[\ell]$ for
  all $k\in\N$, then by \eqref{cnd:divsim} there exists
  $\statet[\ell+1]$ such that
  $\statet[\ell]\step{\silent}\statet[\ell+1]$ and
  $\states[\sigma(\ell)+k]\bbisimd\statet[\ell+1]$ for some $k\in\N$;
  we can then define $\sigma(\ell+1)=k$.

  Otherwise, there exists some $k\in\N$ such that
  $\states[\sigma(\ell)+k]\bbisimd\statet[\ell]$ and
  $\states[\sigma(\ell)+k+1]\not\bbisimd\statet[\ell]$. Since
  $\states[\sigma(\ell)+k]\step{\silent}\states[\sigma(\ell)+k+1]$ it
  follows by \eqref{cnd:stepsim} that there exist $\statet[\ell]''$
  and $\statet[\ell+1]$ such that
  $\statet[\ell]\ssteps\statet[\ell]''\step{\opt{\silent}}\statet[\ell+1]$,
  $\states[\sigma(\ell)+k]\bbisimd\statet[\ell]''$ and
  $\states[\sigma(\ell)+k+1]\bbisimd\statet[\ell+1]$. Clearly, we have that
  $\statet[\ell]\neq\statet[\ell+1]$, so
  $\statet[\ell]\plussteps\statet[\ell+1]$ and we can define $\sigma(\ell+1)=\sigma(\ell)+k+1$.
\end{itemize}
From the existence of an infinite sequence of states
$(\statet[\ell])_{\ell\in\N}$ such that $\statet=\statet[0]$ and
$\statet[\ell]\step{\silent}\statet[\ell+1]$ it follows that $\statet$
is divergent, as was to be shown.
\end{proof}

As the following example illustrates, however, a symmetric binary relation on $\States$ relating
states that satisfies \eqref{cnd:stepsim} of Definition~\ref{def:bbisimd} and relates divergent states to divergent
states only is \emph{not necessarily} included in a divergence-preserving
branching bisimulation relation. In other words, a symmetric binary
relation on $\States$ that satisfies \eqref{cnd:stepsim} and only
relates divergent states to divergent states may relate states that
are not divergence-preserving branching bisimilar.

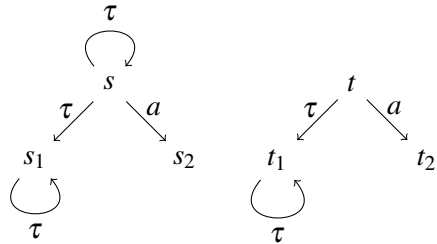
\begin{figure}[htb]
    \begin{center}
      \begin{tikzpicture}[node distance=40pt]
        \node (s) {$\states$};
        \node [below left of=s] (s1) {$\states[1]$};
        \node [below right of=s] (s2) {$\states[2]$};
        \node [right=80pt of s] (t) {$\statet$};
        \node [below left of=t] (t1) {$\statet[1]$};
        \node [below right of=t] (t2) {$\statet[2]$};
   
        \path[->] (s) edge[loop, min distance=.5cm, in=50,out=130,looseness=5] node[above] {$\silent$} (s);
        \path[->] (s) edge node[near start,left] {$\silent$} (s1);
        \path[->] (s1) edge[loop, min distance=.5cm, out=230,in=310,looseness=5] node[below] {$\silent$} (s1);
        \path[->] (s) edge node[near start,right] {$\acta$} (s2);
        
        \path[->] (t) edge node[near start,left] {$\silent$} (t1);
        \path[->] (t1) edge[loop, min distance=.5cm, out=230,in=310,looseness=5] node[below] {$\silent$} (t1);
        \path[->] (t) edge node[near start,right] {$\acta$} (t2);
        
      \end{tikzpicture}
    \end{center}
\caption{An example transition system illustrating that \eqref{cnd:divsim} cannot be replaced by the requirement that $\brel$
  relates divergent states to divergent states.}\label{fig:insufficientdc}
\end{figure}
\begin{exmp}
  Consider the transition system depicted in
  Figure~\ref{fig:insufficientdc}.
  The symmetric closure of the relation
  $\brelsym=\{(\states,\statet),(\states[1],\statet[2]),(\states[2],\statet[2])\}$ satisfies
  \eqref{cnd:stepsim} and it relates divergent states to
  divergent states only. It does not, however, satisfy
  \eqref{cnd:divsim}, for $s\brel t$ and defining $\states[k]=s$ for
  all $k\in\N$ we get an infinite sequence of states
  $(\states[k])_{k\in\N}$ such that
  $\states[k]\step{\silent}\states[k+1]$ and $\states[k]\brel\statet$
  for all $k\in\N$, while there does not exist a $\statet'$ such that
  $\statet\plussteps\statet'$ and $\states[k]\brel\statet'$ for some
  $k\in\N$.
  Note that $s$ admits a complete path at which $\acta$ is
  continuously (weakly) enabled, whereas $t$ does not admit such a complete path.
\end{exmp}

\section{Modal characterisations} \label{sec:modchar}

As shown in \cite{Gla93a}, to get an (action-based) modal logic that
is adequate for branching bisimilarity one could take an adaptation of
standard Hennessy-Milner logic replacing, for all actions
$\act\in\Acts$ in the usual unary may and must modalities
$\may{\acta}$ and $\must{\acta}$ by a binary \emph{just-before}
modality $\jb{\acta}$. A state $\states$ satisfies the formula
$\varphi\jb{\acta}\psi$ if, and only if, there exist states
$\states''$ and $\states'$ such that
$\states\ssteps\states''\step{\opt{\act}}\states'$, $\varphi$ holds in
$\states''$ and $\psi$ holds in $\states'$. To get an adequate logic
for divergence-preserving branching bisimilarity, it suffices to add a
unary \emph{divergence modality} $\Div{}$ such that $\states$
satisfies $\Div{\varphi}$ if, and only if, there exists an infinite
sequence of states $(\states[k])_{k\in\N}$ such that
$\states\ssteps\states[0]$, $\states[k]\step{\silent}\states[k+1]$ and
$\varphi$ holds in $\states[k]$ for all $k\in\N$.

Let $\Phi$ be the class of formulas generated by the following grammar:
  \begin{equation*}
    \varphi ::=
      \neg\varphi\ \mid\
      \bigwedge\Phi'\ \mid\
      \varphi\jb{\act}\varphi\ \mid\
      \Div\varphi\qquad (\act\in\Acts,\ \varphi\in\Phi,\
      \Phi'\subseteq\Phi)
  \enskip.
\end{equation*}
We then have that states $\states$ and $\statet$ are divergence-proving branching bisimilar if, and only if, $\states$ and
$\statet$ satisfy exactly the same formula in $\Phi$ \cite{GLT09a}. We
may restrict the cardinality of $\Phi'$ in conjunctions to the
cardinality of the set of states $\States$.
\begin{exmp}
  Consider again the transition system depicted in
  Figure~\ref{fig:insufficientdc}. States $\states$ and $\statet$ are
  not divergence-preserving branching bisimilar.
  The formula
  $\Div\left(\top\jb{\acta}\top\right)$ (in which $\top$ abbreviates $\bigwedge\emptyset$)
  expresses the existence of a divergence on which the action
  $\acta$ is continuously enabled. It is satisfied by state $\states$,
  but not by $\statet$.
\end{exmp}

There is also an intuitive correspondence between branching
bisimilarity and the state-based temporal logic CTL$^{*}_{-X}$
(CTL$^{*}$ without the next-state modality) \cite{EC82}. The standard
semantics of CTL$^{*}_{-X}$ is, however, with respect to Kripke structures, in
which states rather than transitions have labels and the transition
relation is assumed to be total. 
To formalise the correspondence, De Nicola and Vaandrager
devised a framework of translations between labelled transition systems
and Kripke structures \cite{DV95}. The main idea of the translation from labelled
transition systems to Kripke structures is that
\begin{enumerate}
  \item every transition $\states\step{\acta}\statet$
    ($\acta\neq\silent$) is replaced by two transitions
    $\states\step{}\statet_{\acta}$ and
    $\statet_{\acta}\step{}\statet$, where $t_{\acta}$ is a fresh
    state that is labelled with $\{\acta\}$;
  \item every transition $\states\step{\silent}\statet$ gives rise to
    a transition $\states\step{}\statet$; and
  \item\label{item:totalise} for every state $\states$ without outgoing transitions (i.e.,
    every deadlock state of the labelled transition system) a
    transition $\states\step{}\states$ is added to satisfy the
    totality requirement of Kripke structures.
  \end{enumerate}

  \begin{figure}[htb]
     \begin{center}
      \begin{tikzpicture}[node distance=40pt]
        \node (s) {$\states$};
        \node [below left of=s] (s1) {$\states[1]$};
        \node [below right of=s] (s2) {$\states[2,\acta]$};
        \node [below of=s2] (s2a) {$\states[2]$};
        \node [right=80pt of s] (t) {$\statet$};
        \node [below left of=t] (t1) {$\statet[1]$};
        \node [below right of=t] (t2) {$\statet[2,\acta]$};
        \node [below of=t2] (t2a) {$\statet[2]$};
        \node [above right of=s2, node distance=5mm] (s2lab) {$\{\acta\}$};
        \node [above right of=t2, node distance=5mm] (t2lab) {$\{\acta\}$};
   
        \path[->] (s) edge[loop, min distance=.5cm, in=50,out=130,looseness=5] (s);
        \path[->] (s) edge (s1);
        \path[->] (s1) edge[loop, min distance=.5cm, out=230,in=310,looseness=5] (s1);
        \path[->] (s) edge (s2);
        \path[->] (s2) edge (s2a);
        \path[->] (s2a) edge[loop, min distance=.5cm, out=230,in=310,looseness=5] (s2a);
       
        \path[->] (t) edge (t1);
        \path[->] (t1) edge[loop, min distance=.5cm, out=230,in=310,looseness=5] (t1);
        \path[->] (t) edge (t2);
        \path[->] (t2) edge (t2a);
        \path[->] (t2a) edge[loop, min distance=.5cm, out=230,in=310,looseness=5] (t2a);
        
      \end{tikzpicture}
    \end{center}
\caption{Result of apply De Nicola and Vaandrager's translation to the
  labelled transition system in Figure~\ref{fig:insufficientdc}.}\label{fig:KSDNV}
 \end{figure}
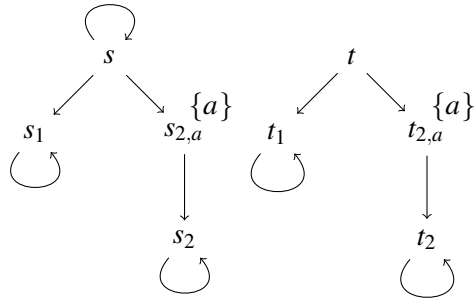
  
  \begin{exmp}
     If we apply the translation sketched above to the labelled
     transition system depicted in Figure~\ref{fig:insufficientdc}, then we
     get the Kripke structure depicted in Figure~\ref{fig:KSDNV}. Note
     that by clause~\ref{item:totalise} of the translation state
     $\states[2]$ gets a transition to itself, whereas it is a
     deadlock state in the orginal transition system.
     Clearly, there is no CTL$_{-X}$ formula that distinguishes, e.g.,
     between $s_1$ and $s_2$, although in the labelled transition
     system depicted in Figure~\ref{fig:insufficientdc} these states
     are not divergence-preserving branching bisimilar.
   \end{exmp}
   
De Nicola and Vaandrager propose a notion of \emph{divergence-sensitive
  branching bisimilarity} on finite LTSs and establish that two states
in an LTS are divergence-sensitive branching bisimilar if, and only
if, in the Kripke resulting from the translation sketched above they
satisfy the same CTL$^{*}_{-X}$ formulas. Divergence-sensitive branching
bisimilarity coincides with divergence-preserving branching
bisimilarity on deadlock-free LTSs. In fact, the only difference
between divergence-sensitive branching bisimilarity and
divergence-preserving branching bisimilarity is that the latter
distinguishes between deadlock and livelock states, whereas the former
does not.

To preserve the distinction between deadlock and livelock, a modified translation is proposed in \cite{GLT09b}, obtained from
the translation sketched above by replacing clause~\ref{item:totalise}
by
\begin{enumerate}
  \addtocounter{enumi}{2}
  \renewcommand{\labelenumi}{\theenumi$'$.}
  \item add a fresh state $d$ labelled with $\{\delta\}$, and
    for every state $\states$ without outgoing transitions a transition $\states\step{}d$.
\end{enumerate}

 \begin{figure}[htb]
     \begin{center}
      \begin{tikzpicture}[node distance=40pt]
        \node (s) {$\states$};
        \node [below left of=s] (s1) {$\states[1]$};
        \node [below right of=s] (s2) {$\states[2,\acta]$};
        \node [below of=s2] (s2a) {$\states[2]$};
        \node [right=80pt of s] (t) {$\statet$};
        \node [below left of=t] (t1) {$\statet[1]$};
        \node [below right of=t] (t2) {$\statet[2,\acta]$};
        \node [below of=t2] (t2a) {$\statet[2]$};
        \node [above right of=s2, node distance=5mm] (s2lab) {$\{\acta\}$};
        \node [above right of=t2, node distance=5mm] (t2lab) {$\{\acta\}$};
        \node [below of=s2a, node distance=20pt] (s2adown){};
        \node [right=36pt of s2adown] (d) {$d$};
         \node [above of=d, node distance=5mm] (dlab) {$\{\delta\}$};
  
        \path[->] (s) edge[loop, min distance=.5cm, in=50,out=130,looseness=5] (s);
        \path[->] (s) edge (s1);
        \path[->] (s1) edge[loop, min distance=.5cm, out=230,in=310,looseness=5] (s1);
        \path[->] (s) edge (s2);
        \path[->] (s2) edge (s2a);
        \path[->] (d) edge[loop, min distance=.5cm,
        out=230,in=310,looseness=5] (d);
        \path[->] (s2a) edge (d);
        \path[->] (t2a) edge (d);
       
        \path[->] (t) edge (t1);
        \path[->] (t1) edge[loop, min distance=.5cm, out=230,in=310,looseness=5] (t1);
        \path[->] (t) edge (t2);
        \path[->] (t2) edge (t2a);
        
      \end{tikzpicture}
    \end{center}
\caption{Result of apply the deadlock preserving translation to the
  labelled transition system in
  Figure~\ref{fig:insufficientdc}.}\label{fig:KSDNVDT}
 \end{figure}
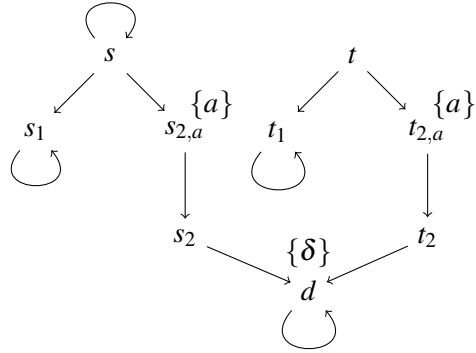

  \begin{exmp}
     Applying the modified translation on the labelled transition
     in Figure~\ref{fig:insufficientdc}, we 
     get the Kripke structure in Figure~\ref{fig:KSDNVDT}.
     Note that $\states[1]$ does not satisfy the CTL$^{*}_{-X}$ formula
     $\mathsf{EF}\, \delta$, while $\states[2]$ does.
   \end{exmp}

Two states in a labelled transition system are divergence-preserving
branching bisimilar if they satisfy the same CTL$^{*}_{-X}$ formulas in
the Kripke structure that results from the modified transition
\cite{GLT09b}.

\section{Congruence} \label{sec:congruence}

An important reason to prefer divergence-preserving branching
bisimilarity over divergence-sensitive branching bisimilarity is that
the former is compatible with parallel composition, whereas the latter
is not.

\begin{figure}[htb]
  \begin{center}
    \begin{tikzpicture}[node distance=20pt]
      \node (s1t) {$\states[1]\mathbin{\|}\statet$};
      \node [below of=s1t, node distance=40pt] (s2t) {$\states[2]\mathbin{\|}\statet$};
      \node [right of=s1t,anchor=west] (E1) {};
      \node [right of=E1] (t) {$\statet$};
      \node [right of=t] (P1) {};
      \node [right of=P1](s1) {$\states[1]$};
      \node [below of=s1, node distance=40pt] (s2) {$\states[2]$};
      \node [right of=s1] (P2) {};
      \node [right of=P2] (t') {$\statet'$};
      \node [right of=t'] (E2) {};
      \node [right of=E2, anchor=west] (s1t') {$\states[1]\mathbin{\|}\statet'$};
      \node [below of=s1t',node distance=40pt] (s2t') {$\states[2]\mathbin{\|}\statet'$};

      
      \path[->] (t') edge[loop, min distance=.5cm,out=230,in=310,looseness=5] node[below]{$\silent$} (t');
      \path[->] (s1t) edge node[left] {$\acta$} (s2t);
      \path[->] (s1) edge node[left] {$\acta$} (s2);
      \path[->] (s1t') edge node[left] {$\acta$} (s2t');
      \path[->] (s1t') edge[loop, min distance=.5cm,out=20,in=340,looseness=5] node[right]{$\silent$} (s1t');
      \path[->] (s2t') edge[loop, min distance=.5cm,out=20,in=340,looseness=5] node[right]{$\silent$} (s2t');
    \end{tikzpicture}
 \end{center}\caption{Divergence-sensitive branching bisimilarity is
   not compatible with parallel composition.} \label{fig:parcompat}
\end{figure}
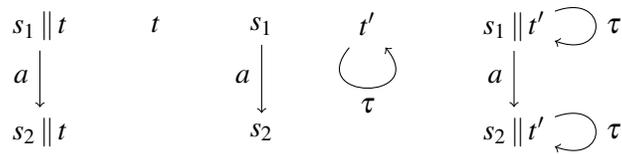

\begin{exmp}
  Consider the transition system in Figure~\ref{fig:parcompat}.
  States
    $\states[1]\mathbin{\|}\statet$ and
    $\states[2]\mathbin{\|}\statet$
  represent the parallel compositions of states
  $\states[1]$ and $\statet$, and of states $\states[2]$ and $\statet$, respectively.
  Similarly, states
    $\states[1]\mathbin{\|}\statet'$ and
    $\states[2]\mathbin{\|}\statet'$
  represent the parallel compositions of states $\states[1]$ and
  $\statet'$, and of states $\states[2]$ and $\statet'$, respectively.
  Recall that divergence-sensitive branching bisimilarity does not
  distinguish deadlock (state $\statet$) and livelock (state
  $\statet'$), so we have that $\statet$ and $\statet'$ are
  divergence-sensitive branching bisimilar. States
  $\states[1]\mathbin{\|}\statet$ and $\states[1]\mathbin{\|}\statet'$
  are, however, not divergence-sensitive branching bisimilar. Note that
  $\states[1]\mathbin{\|}\statet'$ has a complete path on which
  $\acta$ is continuously enabled, whereas
  $\states[1]\mathbin{\|}\statet$ does not have such a complete path,
  and so these two states do not satisfy the same CTL$^{*}_{-X}$ formulas.
\end{exmp}

Divergence-preserving branching bisimilarity is the coarsest
equivalence included in divergence-sensitive branching bisimilarity
that is compatible with parallel composition \cite{GLT09b}. Hence, it is also the
coarsest congruence for parallel composition relating only processes
that satisfy the same CTL$^{*}_{-X}$ formulas.

It is well-known that branching bisimilarity is not compatible with
non-deterministic choice, and that the coarsest behavioural
equivalence that is included in branching bisimilarity and that is
compatible with non-deterministic choice, is obtained by adding a
so-called \emph{root condition}. The same holds for divergence-preserving
branching bisimilarity.
\begin{defn}
  Let $\brel$ be a divergence-preserving branching bisimulation. We
  say that $\brel$ satisfies the \emph{root condition} for $\states$
  and $\statet$ if, whenever
 \begin{enumerate}\itemsep 0pt
  \renewcommand{\labelenumi}{(R\theenumi)}
  \item \label{cnd:root1}
    if $\states\step{\act}\states'$ for some state $\states'$, then there exists a state $\statet'$
    such that
      \plat{$\statet\step{\act}\statet'$}
    and
    $\states'\brel\statet'$.
  \item \label{cnd:root2}
    if $\statet\step{\act}\statet'$ for some state $\statet'$, then there exists a state $\states'$
    such that
      \plat{$\states\step{\act}\states'$}
    and
    $\states'\brel\statet'$.
  \end{enumerate}
  States $\states$ and $\statet$ are \emph{rooted
    divergence-preserving branching bisimilar} if there is a
  divergence-preserving branching bisimulation relation $\brel$
  satisfying the root condition for $\states$ and $\statet$ such that
  that $\states\brel\statet$.
\end{defn}

In \cite{FvGL2019}, formats for transition system
specifications are presented that guarantee that divergence-preserving
branching bisimilarity and its rooted variant are compatible with the
operators defined by the transition system specification. These
formats relax the requirements of the branching bisimulation and
rooted branching bisimulation formats of \cite{FGW12}. The relaxation
of the formats is meaningful: the process-algebraic operations for
\emph{priority} \cite{BBK86} and \emph{sequencing}
\cite{Blo94,BLB19,BLY17}, with which (rooted) branching bisimilarity
is \emph{not} compatible, are in the rooted divergence-preserving
branching bisimulation format. So, in contrast to its
divergence-insensitive variant, rooted divergence-preserving branching
bisimilarity is compatible with priority and sequencing.

The structural operational rule for the recursion
operator $\mu X.\_$, which was considered in the context of
observation equivalence by Milner \cite{Mil89} and in the context of
divergence-sensitive variants of observation equivalence by Lohrey,
D'Argenio and Hermanns \cite{LDH05}, is not in the format for rooted
divergence-preserving branching bisimilarity. Nevertheless, rooted
divergence preserving branching bisimilarity is compatible also with this
operator \cite{GLS20}. The proof of this fact requires an adaption of
the up-to technique used by Milner in his argument that (strong)
bisimilarity is compatible with recursion \cite{Mil89h}.

\section{Expressiveness of process calculi} \label{sec:applications}

Phillips showed that abstraction from divergence can be exploited to prove that every
recursively enumerable transition system is branching bisimilar to a
boundedly branching computable transition system
\cite{Phi93}\footnote{Phillips actually claimed the correspondence
  modulo observation equivalence, but it is easy to see that his proof
  also works modulo branching bisimilarity.}. In contrast, there exist
recursively enumerable transition systems that are not
divergence-preserving branching bisimilar to a computable transition
system (cf., e.g.,  Example 3.6 in \cite{BLT13}).
Hence, in a theory that aims to integrate computability and
concurrency, divergence preservation is important.

In \cite{BLT13}, interactivity is added to Turing machines by
associating an action with every computation step. This so-called
\emph{reactive} Turing machine has a transition system semantics and
can be studied from a concurrency-theoretic perspective. A transition
system is called \emph{executable} if it is behaviourally equivalent
to the transition system associated with a reactive Turing
machine. The notion of executability provides a way to characterise the
absolute expressiveness of a process calculus. If every transition
system that can be specified in the calculus is executable, then the
calculus is said to be executable. Conversely, if every executable
transition system can be specified in the calculus, then the calculus
is said to be behaviourally complete.

A calculus with
constants for deadlock and successful termination, unary action
prefixes, binary operations for non-deterministic choice, sequencing
and ACP-style parallel composition, iteration and nesting is both
executable and behaviourally complete up to divergence-preserving
branching bisimilarity \cite{BLY17}. The $\pi$-calculus is also
behaviourally complete up to divergence-preserving branching
bisimilarity. Since it allows the specification of transition systems
with unbounded branching, it is, however, not executable up to
divergence-preserving branching bisimilarity; it is nominally
orbit-finitely executable up to the divergence-insensitive variant of branching
bisimilarity \cite{LY20}.

The aforementioned results illustrate the role of divergence in the
consideration of the absolute expressiveness of process
calculi. Preservation of divergence is also widely accepted as an important
criterion when comparing the relative expressiveness of process
calculi \cite{Pet19}.

\section{Conclusions} \label{sec:conclusions}

We have discussed ther relational and modal characterisations of
divergence-preserving branching bisimilarity, commented on its
compatibility with respect to process algebraic operations and on its
role in the study of the absolute expressiveness. We conclude by briefly mentioning some directions for future work.

Sound and complete axiomatisations for the divergence-sensitive
spectrum of observation congruence for basic CCS with recursion are
provided in \cite{LDH05}. The congruence result in \cite{GLS20} can
serve as a stepping stone for providing similar sound and complete
axiomatisations for divergence-preserving branching
bisimilarity. Then, it would also be interesting to consider the
axiomatisation of divergence-preserving branching bisimilarity for
full CCS with recursion, although that would first require a
non-trivial extension of the congruence result.

Ad hoc up-to techniques for divergence-preserving branching
bisimilarity have already been used, e.g., in the congruence proof in
\cite{GLS20} and in proof that the $\pi$-calculus is behaviourally
complete \cite{LY20}. Recently, several more generic up-to
techniques for branching bisimilarity were proved sound
\cite{ERL20}. An interesting direction for future work would be to
consider extending those up-to techniques for divergence-preserving
branching bisimilarity too.

\bibliographystyle{eptcs}
\bibliography{dpbb}
\end{document}